\documentclass[12pt]{article}
\usepackage[utf8]{inputenc}
\usepackage[margin=0.75in]{geometry}
\usepackage{booktabs, float}
\usepackage{tabularx}
\usepackage{hyperref}
\hypersetup{
	colorlinks=true,
	linkcolor=blue,
	citecolor=blue,
    urlcolor=blue
}
\usepackage{amsmath, amsfonts, amssymb, amsthm}
\usepackage{enumerate}

\newtheorem{theorem}{Theorem}

\usepackage{natbib}
\usepackage{algpseudocode}
\usepackage[ruled,vlined,linesnumbered]{algorithm2e}
\usepackage{appendix}
\usepackage{amssymb,amsmath,relsize}
\usepackage{graphicx} 
\usepackage{subfigure}
\usepackage{svg}
\usepackage{svg-extract}
\usepackage{rotating}
\usepackage{adjustbox}
\usepackage{csquotes}
\usepackage{mathrsfs}
\usepackage{isomath}
\usepackage{lscape}
\usepackage{multirow}
\usepackage{tikz}
\usetikzlibrary{arrows.meta, positioning}
\begin{document}
	
	\title{Optimum Multiple Sampling Plan Based on the Process Capability Index $C_{py}$ Under Type-II Hybrid Censoring}



	\author{Rajat Das$^1$, Tanmay Kayal$^2$, Yogesh Mani Tripathi$^1$, and Tanmay Sen$^3$\thanks{Corresponding author: tanmay.sen@isical.ac.in}}
	\date{}
	{\footnotesize  \maketitle \noindent {\it $^{1}$Department of Mathematics, Indian Institute of Technology Patna, Patna, 801106, Bihar, India}\\
    {\it $^2$School of Mathematical \& Statistical Sciences, Indian Institute of Technology Mandi, Kamand, 175075, Himachal Pradesh, India }\\
    {\it $^3$Statistical Quality Control and Operations Research Unit, Indian Statistical Institute, Kolkata, 700108, West Bengal, India}} 
	\begin{abstract}
		\noindent  

        This paper proposes a stage independent multiple sampling plan (SIMSP) to improve inspection efficiency by reducing the number of samples required at each sampling stage. Unlike conventional multiple sampling plans (MSP), the proposed SIMSP eliminates the dependence of each sampling stage on the outcome of the preceding inspection. The proposed approach is developed for non-repairable products sold under a pro-rata warranty policy based on the generalized process capability index $C_{py}$.
        The SIMSP is designed under a Type-II hybrid censoring scheme (Type-II HCS), which provides greater flexibility in controlling test time and failure information in life testing experiments. The asymptotic distribution of the process capability index estimate is used to compute the operating characteristic (OC) function, and the exact Fisher information matrix (FIM) is obtained for further statistical analysis. A constraint optimization problem is formulated to determine the optimal design by minimizing the total cost subject to the manufacturer's and consumer's risk.
        Numerical investigations are conducted to examine the effects of model parameters, warranty policy characteristics, and cost factors on the optimal solution. The results demonstrate that the proposed approach provides an economically efficient and feasible method for lot acceptance while satisfying the tolerable risks requirements.

    		
        ~~
	\end{abstract}
	
\section{Introduction}\label{sec-1}
Acceptance sampling is a basic method of quality control in statistics for determining whether a production lot meets the required quality criteria. It offers an ideal compromise between not conducting an inspection at all and inspecting the entire production lot, which may be costly, impractical, or even destructive to the items. By inspecting a sample from the production lot, the manufacturer or buyer can strike a balance between costs and control over the risks to producers and consumers. Among the available acceptance sampling procedures, the single sampling plan (SSP) is the simplest and most widely adopted because of its straightforward implementation. Under an SSP, a single sample is drawn from the submitted lot, and the acceptance or rejection decision is made immediately based on a predetermined criterion. Its simplicity makes it attractive in practice; however, the entire decision relies on a single sample (see \citep{wang2025constructing}, \citep{wang2022integrated}, and \citep{wu2021acceptance}). Consequently, SSP often requires relatively large sample sizes to achieve specified producer's and consumer's risks, resulting in increased inspection costs.

The double sampling plan (DSP) is an extension of the SSP to increase its efficiency. In the DSP, a second sample is taken only when the information from the first sample is insufficient to accept or reject the lot. In this case, the number of samples needed usually decreases, since many samples of either very high or very low quality can be categorized based on the first sample. \citep{sommers1981two} initially provided the design plan parameters for the two-point variable double-sampling plans, in which the OC curve for the plan was then determined using an approximate method based on information collected from both sample selections. Recent work on DSP can be found in the studies \citep{sridevi2023determination}, \citep{liu2024efficient}, and many others. To make conventional double sampling plans easier to employ, \citep{arizono2020variable} proposed the stage independent double sampling plan (SIDSP), in which the first stage did not influence the second-stage decision. This process simplified the design of variable double-sampling plans since there was no need to compute conditional probability in the OC function.

To further reduce inspection effort, the multiple sampling plan (MSP) extends DSP by allowing several successive sampling stages. Since only a relatively small sample is inspected at each stage, MSP generally achieves a lower average sample number (ASN) than SSP and DSP, particularly when the submitted lot is either of very high or very low quality. For this reason, MSP has received considerable attention in acceptance sampling research. Despite its economic advantages, conventional MSP is mathematically complicated because the decision at each stage depends on the outcomes of all previous stages. Consequently, deriving the OC function becomes increasingly difficult as the number of sampling stages increases. To overcome the computational complexity of conventional MSP, \citep{wu2023stage} proposed the SIMSP as an extension of the SIDSP based on the process capability index $C_{pk}$. In the SIMSP, the inspection outcome at each sampling stage is assumed to be independent of previous stages. This relaxation substantially simplifies the derivation of the OC function while retaining the principal advantage of MSP, namely a reduction in the ASN. Furthermore, SIMSP encompasses SSP and DSP as special cases when the maximum allowable number of stages is $1$ and $2$, respectively.

With increasing competition in the marketplace, it is no longer enough to have products that provide satisfactory quality to the consumers. There is a growing trend among companies to implement warranty programs as a marketing tool to improve consumer satisfaction and enhance their competitiveness. In addition to guaranteeing product quality, the warranty program reduces some of the risks consumers may face from product failures after purchase. Thus, the cost of warranty programs is becoming an important part of production and quality management. Among various warranty policies, the general rebate warranty (GRW) has attracted considerable attention because it provides customers with a predetermined rebate when a product fails within the warranty period, rather than requiring full replacement or repair. Consequently, several reliability acceptance sampling plans have been developed for products marketed under GRW, where warranty cost is explicitly incorporated into the sampling plan design. Interested readers can find recent work on a warranty based acceptance sampling plan in the literature, including \citep{chakrabarty2020optimum}, \citep{das2025bayesian}, \citep{chakrabarty2021optimum}, and     \citep{wang2024modified}.

Since acceptance decisions for highly reliable products are often based on life testing experiments, the choice of an appropriate censoring scheme plays a crucial role in balancing experimental cost and testing time. Conventional Type-I censoring may terminate the experiment before a sufficient number of failures are observed, whereas Type-II censoring often requires an excessively long testing duration. Type-II HCS combines the advantages of both schemes, thereby providing greater flexibility and practical efficiency in life-testing experiments. Owing to these advantages, Type-II HCS has been widely adopted in reliability analysis and acceptance sampling; see \citep{bhattacharya2015computation}, \citep{sen2018statistical}, \citep{salah2021statistical}, and others.

Although SIMSP based on process capability indices have demonstrated superior inspection efficiency by reducing average sample sizes, their application has been largely limited to conventional process capability indices and has overlooked the influence of warranty policies. Moreover, existing warranty-based acceptance sampling plans rely predominantly on lifetime characteristics rather than process capability measures. Although numerous acceptance sampling plans have been developed based on classical process capability indices, such as $C_p, C_{pk}, C_{pc}$, etc., existing studies have primarily focused on these conventional indices and their variants. By comparison, integration of the generalized process capability index (GPCI) $C_{py}$, introduced in \citep{maiti2010generalizing}, for acceptance sampling plans has received little attention. Consequently, there remains a need to develop a sampling plan framework that leverages the greater flexibility and broader applicability of generalized capability indices. To the best of our knowledge, no existing study has proposed a sampling plan for products sold under a pro-rata warranty policy using the generalized process capability index $C_{py}$. This gap motivates the development of the proposed sampling plan, which simultaneously accounts for process capability, inspection, and warranty-related costs. Accordingly, this paper develops a generalized process capability index based, stage independent multiple sampling plan for products sold under a pro-rata warranty policy. The proposed plan determines the optimal design parameters by minimizing the expected total cost while satisfying the producer's and consumer's risk requirements. The performance of the proposed methodology is investigated through numerical studies and comparisons with existing sampling schemes.

The rest of the paper is structured as follows. In section \ref{sec-2}, we give a brief description of the generalized process capability index $C_{py}$ and provide a brief review of the baseline Chen distribution. In section \ref{sec-3}, the estimation procedure for $C_{py}$ is discussed, and the exact Fisher information matrix is derived. The operating procedure of the sampling plan, the development of the OC function under the asymptotic distribution of $\hat{C}_{py}$, and the optimization model to obtain the optimal parameters of the sampling plan under the warranty rebate policy are presented in section \ref{sec-4}. Numerical examples and comparative studies, along with tables of optimal parameters for the proposed plan, are provided in section \ref{sec-5}. The practical applicability of the proposed approach is demonstrated through a numerical example and a sensitivity analysis in Section \ref{sec-6}. Section \ref{sec-7} concludes the paper with a summary of the main findings and future research directions.

\section{The index \texorpdfstring{$C_{py}$}{Cp} and the Chen distribution}\label{sec-2}
This study considers the GPCI, denoted by $C_{py}$, which was originally introduced in \cite{maiti2010generalizing} as a generalization of traditional process capability measures. A key strength of this index is its versatility: it accommodates both continuous and discrete quality characteristics and applies to processes following either normal or non-normal distributions. Let $F(\cdot)$ represent the CDF of the process. The GPCI thus offers a comprehensive and adaptable framework for assessing process performance, particularly in scenarios involving asymmetric tolerance limits or non-normal behavior. The mathematical formulation of the index is given by
\begin{equation}
    \begin{aligned}
        C_{py} & =\frac{F(U)-F(L)}{F(UDL)-F(LDL)}=\frac{p}{p_0}
    \end{aligned}
\end{equation}
Here, the symbols $L$ and $U$ correspond to the lower and upper specification limits, whereas $LDL$ and $UDL$ denote the lower and upper desirable (or tolerance) limits, as commonly used by practitioners. The index $C_{py}$ possesses several notable features. It is applicable under both one-sided and two-sided specification settings. 

The process capability index $C_{py}$ has several desirable properties. When the process yield equals the desired yield ($p=p_0$; for a normal process), $C_{py}=1$. If $p>p_0$, then $C_{py}>1$, indicating a capable process, whereas $C_{py}<1$ for $p<p_0$. As $p \to 0$, $C_{py}\to 0$, making zero its lower bound. Thus, $C_{py}$ provides a direct and intuitive measure of process capability relative to the minimum acceptable yield $p_0$.
The index applies to both normal and non-normal distributions, continuous and discrete quality characteristics, and unilateral as well as bilateral specification limits. Owing to its direct relationship with process yield, $C_{py}$ is easy to interpret and implement in practice. Moreover, its estimator has relatively simple distributional properties, facilitating statistical inference. Some recent works on GPCI can be found in \cite{dey2018bootstrap}, \citep{kayal2026inference}, and \cite{kumar2022parametric}.

The Chen distribution (CD), introduced by \citep{chen2000new}, is a two-parameter lifetime distribution that is useful for fitting increasing or bathtub-shaped hazard rate functions. It is worth mentioning that, unlike other lifetime models, such as Exponential and Weibull distributions, whose hazard rate functions are usually constant or bathtub, but may not be suitable to produce a good bathtub shape of the failure rates. The Chen distribution appropriately reflects the three periods in a product's lifecycle: the early failure period, the useful life period, and the wear-out failure period. The Chen distribution has many desirable statistical properties, such as positive skewness, analytical forms of the cumulative distribution function and hazard rate function, and manageable maximum-likelihood estimation. Due to these appealing properties, the Chen distribution has found applications (see \citep{singh2026inference}, \citep{dey2026statistical}, and \citep{zhang2024parameter}) in reliability, survival analysis, maintenance, warranty, and quality control, and has inspired many generalizations of the lifetime distribution.
The cumulative distribution function (CDF) and the corresponding probability density function (PDF) of the Chen distribution are of the form
\begin{equation}
    F(x;\eta,\lambda) = 1-e^{\eta(1-e^{x^{\lambda}})},\quad x>0,\eta>0,\lambda>0,
\end{equation}
and 
\begin{equation}\label{pdf}
    f(x;\eta,\lambda) = \eta\lambda x^{\lambda-1}e^{x^\lambda+\eta(1-e^{x^{\lambda}})},\quad x>0,\eta>0,\lambda>0    
\end{equation}
respectively, where $\eta$ and $\lambda$ are unknown parameters. The parameter $\lambda$ controls the shape of the hazard rate function (HRF). Specifically, the hazard rate is increasing when $\lambda \geq 1$, whereas it exhibits a bathtub shape for $0 < \lambda <1$. Possible shapes of PDF and HRF are presented in the Figure \ref{pdf-hrf}. Now, under the assumption on the $CD(\eta,\lambda)$, the index $C_{py}$ is defined as follows
\begin{equation}
    C_{py} = \frac{e^{\eta\left(1-e^{L^{\lambda}}\right)}-e^{\eta\left(1-e^{U^{\lambda}}\right)}}{p_0}.
\end{equation}
\begin{figure}
    \centering
    \includegraphics[width=1\linewidth]{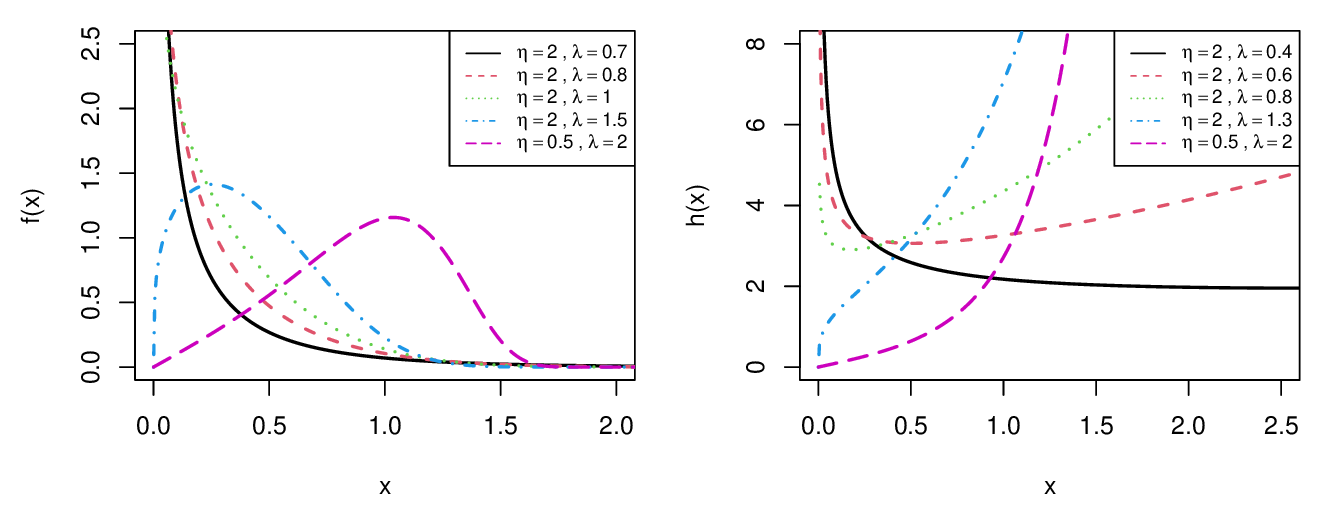}
    \caption{PDF and HRF plot for Chen distribution.}
    \label{pdf-hrf}
\end{figure}
\section{Estimation of \texorpdfstring{$C_{py}$}{Cpy}}\label{sec-3}
Let, $X=(X_{1:n},X_{2:n},\dots,X_{d:n})$ be the Type-II hybrid censored sample of size $d$ from a sample size $n$ from $CD(\eta,\lambda)$ with probability distribution function given \eqref{pdf}. Then the likelihood function can be obtained as
\begin{equation}
    L(\eta,\lambda\mid \underline{x}) = \frac{n}{(n-d)!}\prod_{i=1}^df(x_{i:n}\mid \eta,\lambda)\left[1-F(x_{\tau_0:n}\mid \eta,\lambda\right]^{n-d}.
\end{equation}
Here, $d$ denotes the total number of failures in the experiment up to time $\tau_0$. In other words, we can write
\begin{equation*}
    (d,\tau_0)= \begin{cases}
        (d_1,X_{0}) & \text{if} ~ X_{r:n} < X_{0}, \\
        (r,X_{r:n}) & \text{if} ~ X_{0} < X_{r:n},\\
        (n,X_{0})   & \text{if} ~ X_{n:n} < X_{0}
    \end{cases}
\end{equation*}
where $d_1$ denotes the number of failures that occur before time points $X_0$. Therefore, based on the observed data, the corresponding log-likelihood function $\ell(\eta,\lambda)$, ignoring the additive constant, is given by 
\begin{equation}\label{log-lik}
    \ell(\eta,\lambda) = d\ln{\eta} + d\ln{\lambda} + \lambda \sum_{i=1}^d\ln{x_{i:n}} + \sum_{i=1}^dx_{i:n}^\lambda + \eta\sum_{i=1}^d\left(1-e^{x_{i:n}^\lambda}\right) + (n-d)\eta\left(1-e^{\tau_{0}^\lambda}\right).
\end{equation}
\begin{theorem}\label{th-1}
    For $d\geq1$ and fixed $\lambda>0$, the MLE of $\eta$ exists uniquely and can be expressed as 
    \begin{equation}\label{eta-eq}
        \hat{\eta}(\lambda)=-\frac{d}{\sum_{i=1}^d\left(1-e^{x_{i:n}^\lambda}\right) + (n-d)\left(1-e^{\tau_0^\lambda}\right)}.
    \end{equation}
\end{theorem}
\begin{proof}
    See Appendix \hyperref[appendix-A]{A}.
\end{proof}
Substituting $\eta=\hat{\eta}(\lambda)$, in the log-likelihood Equation \eqref{log-lik}, we obtain the profile log-likelihood for $\lambda$ and it can be written as 
\begin{equation}
    \ell(\lambda)=d\ln{\lambda} + \lambda\sum_{i=1}^d\ln{x_{i:n}} + \sum_{i=1}^dx_{i:n}^\lambda - d\ln\left[-\sum_{i=1}^d\left(1-e^{x_{i:n}^\lambda}\right) - (n-d)\left(1-e^{\tau_0^\lambda}\right)\right].
\end{equation}
\begin{theorem}\label{th-2}
    The MLE $\hat{\lambda}$ of $\lambda$ exists and is the unique solution of the equation $\psi(\lambda)=0$ with
    \begin{equation}\label{lam-eq}
        \psi(\lambda)=\frac{d}{\lambda} + \sum_{i=1}^d\ln{x_{i:n}} + \sum_{i=1}^dx_{i:n}^\lambda\ln{x_{i:n}} + d\frac{\sum_{i=1}^d e^{x_{i:n}^\lambda}x_{i:n}^\lambda\ln{x_{i:n}}+(n-d)e^{\tau_0^\lambda}\tau_0^\lambda\ln{\tau_0}}{\sum_{i=1}^d\left(1-                                                   e^{x_{i:n}^\lambda}\right) + (n-d)\left(1- e^{\tau_0^\lambda}\right)}=0.
    \end{equation}
\end{theorem}
\begin{proof}
    See Appendix \hyperref[appendix-B]{B}.
\end{proof}
From Theorem \ref{th-2}, we see that the closed-form solution of the MLE $\hat{\lambda}$ does not exist from the equation $\psi(\lambda)=0$ in Equation \eqref{lam-eq}. Therefore, an iterative approach, Algorithm \ref{algo-1}, is proposed to compute the MLE of $\lambda$ numerically. In addition, the MLE $\hat{\eta}$ of parameter $\eta$ could be further obtained from Theorem \ref{th-1} as 
    \begin{equation}
        \hat{\eta}=-\frac{d}{\sum_{i=1}^d\left(1-e^{x_{i:n}^{\hat{\lambda}}}\right) + (n-d)\left(1-e^{\tau_0^{\hat{\lambda}}}\right)}.
    \end{equation}
Therefore the MLE $\hat{C}_{py}$ of $C_{py}$ can then be obtained by
\begin{equation}
    \hat{C}_{py} = \frac{e^{\hat{\eta}\left(1-e^{L^{\hat{\lambda}}}\right)}-e^{\hat{\eta}\left(1-e^{U^{\hat{\lambda}}}\right)}}{p_0}.
\end{equation}
\begin{algorithm}[H]
\DontPrintSemicolon
\caption{Iterative estimation of $\hat{\lambda}$}
\label{algo-1}
\SetAlgoLined
\KwIn{Initial value $\lambda^{(0)}$ and tolerance $\varepsilon$}
\KwOut{$\hat{\lambda}$}
Set $l \gets 0$\\
\Repeat{$\big|\lambda^{(l+1)} - \lambda^{(l)}\big| < \varepsilon$}{
    Compute
   {\large$\psi_0(\lambda^{(l)}) = -\frac{d}{\sum_{i=1}^d\ln{x_{i:n}} + \sum_{i=1}^dx_{i:n}^\lambda\ln{x_{i:n}} + d\frac{\sum_{i=1}^de^{x_{i:n}^\lambda}x_{i:n}^\lambda\ln{x_{i:n}}+(n-d)e^{\tau_0^\lambda}\tau_0^\lambda\ln{\tau_0}}{\sum_{i=1}^d\left(1-e^{x_{i:n}^\lambda}\right) + (n-d)\left(1-e^{\tau_0^\lambda}\right)}}$}\\
    Update $\lambda^{(l+1)} \gets \psi_0(\lambda^{(l)})$\\
    $l \gets l + 1$\\
}
\Return {$\hat{\lambda} \gets \lambda^{(l)}$\\}
\end{algorithm}

Next, we compute the exact FIM under the Type-II HCS, which we later use to obtain the optimal sampling plan. The expected Fisher information is given by (see \cite{park2009simple}),
\begin{equation*}
    \mathcal{V}(\theta) = \mathcal{I}_{X_0}(\theta)+\mathcal{I}_{1,\dots,r}(\theta)-\mathcal{I}(\theta),
\end{equation*}
where
\begin{equation*}
    \mathcal{I}_{X_0}(\theta)=n\int_0^{X_0}\left\langle\frac{\partial}{\partial\theta}\ln{h_X(x)}\right\rangle f_X(x)dx, 
\end{equation*}
\begin{equation*}
    \mathcal{I}_{1,\dots,r}(\theta)=\int_0^\infty\left\langle\frac{\partial}{\partial\theta}\ln{h_X(x)}\right\rangle \sum_{i=1}^rf_{i:n}(x)dx, 
\end{equation*}
and 
\begin{equation*}
    \mathcal{I}(\theta)=\int_0^{X_0}\left\langle\frac{\partial}{\partial\theta}\ln{h_X(x)}\right\rangle\sum_{i=1}^rf_{i:n}(x)dx,
\end{equation*}
with $h_X(x)$ and $f_{i:n}(x)$ are representing the hazard function of $X$ and pdf of $X_{i:n}$ respectively and $\langle A \rangle$ denotes the matrix $A\cdot A^{\top}$, where $A^{\top}$ is the transpose of the matrix $A$. For our model, $f_{i:n}(x)$ has the form as
\begin{equation*}
    f_{i:n}(x)=i\binom{n}{i}\eta\lambda x^{\lambda-1}e^{x^\lambda+\eta(n-i+1)\left(1-e^{x^\lambda}\right)}\left(1-e^{\eta\left(1-e^{x^\lambda}\right)}\right)^{i-1}.
\end{equation*}
The derivation of the expected Fisher information matrix is presented in the Appendix \hyperref[appendix-C]{C}.
\citep{park2009simple} has also provided the simplified form of the expected number of failures and the expected failure time as
\begin{equation}
    E[\mathcal{D}]=nF_X(X_0)+r-\sum_{i=1}^rF_{i:n}(X_0),
\end{equation}
and
\begin{equation}
    E[\tau]=X_0 + E[X_{r:n-1}] - \int_0^{X_0}\left(1-F_{r:n-1}(t)\right)dt
\end{equation}
respectively, where $F_{i:n}(x)$ represents the CDF of $X_{i:n}$ with $\mathcal{D}$ and $\tau$ are the realization of random variables $d$ and $\tau_0$ respectively.
\section{{The proposed plan}}\label{sec-4}
\subsection{Design of sampling plans}
In a procurement agreement, the manufacturer and the buyer jointly determine the desired quality standards together with the corresponding producer's and consumer's risks. An acceptance sampling procedure is subsequently developed to meet these predetermined specifications. A widely adopted design criterion is to select the sampling plan such that its OC curve intersects two predetermined quality points, namely the acceptable quality level (AQL) and the limiting quality level (LQL). Specifically, when the submitted lot has a quality level of $C_{py}=C_{AQL}$ (good quality), the probability of accepting the lot should be at least $1-\alpha$, where $\alpha$ is the producer's risk. Conversely, when the lot quality is $C_{py}=C_{LQL}$ (poor quality), the probability of acceptance should not exceed $\beta$, where $\beta$ is the consumer's risk.

Figure \ref{simsp-flowchart} illustrates the flowchart of the proposed SIMSP. The operating procedure is summarized as follows:
\begin{enumerate}[Step 1.]
    \item Specify the contractual design parameters, namely the producer's risk $(\alpha)$, the consumer's risk $(\beta)$, the acceptable quality level $(C_{AQL})$, the limiting quality level $(C_{LQL})$ expressed through $C_{py}$, the upper limit on the number of sampling stages $(m)$, and the corresponding unit costs associated with the life-testing procedure. In addition, let $u$ represent the cumulative number of sampling stages, and set $u=1$.
    \item At the $u$th stage, draw a sample of size $n$ from the submitted lot, conduct the life test with failure number $r$ and censoring time $X_0$, and calculate the estimate $\hat{C}_{py}$.
    \item Check whether $u<m$. If the condition is satisfied, execute Step 4; otherwise, continue with Step 5.
    \item Compare $\hat{C}_{py}$ with the two decision limits, i.e. acceptance $(k_a)$ and rejection $(k_r)$. Accept the lot if $\hat{C}_{py}\geq k_a$, reject the lot if $\hat{C}_{py} < k_r$, and continue to another sampling stage if $k_r\leq\hat{C}_{py}<k_a$. In this case, set $u=u+1$ and return to Step 2.
    \item At the final stage, accept the lot if $\hat{C}_{py}\geq k_a$; otherwise, reject the lot.
\end{enumerate}

Following the operating procedure presented above, the probability functions necessary for deriving the OC function of the proposed sampling plan are defined as follows.
 Under some regularity conditions (see \citep{lehmann1998theory}) the statistic $\hat{C}_{py}$ is approximately normal with $E[\hat{C}_{py}]=C_{py}$ and $Var[\hat{C}_{py}]=\Delta_1^2\mathcal{V}_{11}+2\Delta_1\Delta_2\mathcal{V}_{12}+\Delta_2^2\mathcal{V}_{22}$, where $(\Delta_1,\Delta_2)=\left(\frac{\partial \hat{C}_{py}}{\partial\eta},\frac{\partial \hat{C}_{py}}{\partial\lambda}\right)$, and $\mathcal{V}_{11},\mathcal{V}_{12},\mathcal{V}_{22}$ are the elements of the matrix $\mathcal{V}^{-1}(\theta)$. Note that the variance $Var[\hat{C}_{py}]$ of $\hat{C}_{py}$ is derived using the delta method. Therefore, for a submitted lot with quality level $C_{py}=c$, the probabilities of acceptance, rejection, and resampling at the $u$th sampling stage, denoted by $P_a(c), P_r(c)$, and $P_S(c)$, respectively, are defined as follows:

\begin{align}
    P_a(c) = & P(\hat{C}_{py}\geq k_a\mid C_{py} = c)\nonumber\\
           = & 1 - \Phi\left(\frac{k_a-c}{\sqrt{\Delta_1^2\mathcal{V}_{11} +2\Delta_1\Delta_2\mathcal{V}_{12}+\Delta_2^2\mathcal{V}_{22}}}\right), \\
    P_r(c) = & P(\hat{C}_{py}\leq k_r\mid C_{py} = c)\nonumber\\
           = & \Phi\left(\frac{k_r-c}{\sqrt{\Delta_1^2\mathcal{V}_{11} +2\Delta_1\Delta_2\mathcal{V}_{12}+\Delta_2^2\mathcal{V}_{22}}}\right),\\
    P_S(c) = & P(k_r \leq \hat{C}_{py} \leq k_a\mid C_{py} = c) \nonumber\\
           = & \Phi\left(\frac{k_a-c}{\sqrt{\Delta_1^2\mathcal{V}_{11} +2\Delta_1\Delta_2\mathcal{V}_{12}+\Delta_2^2\mathcal{V}_{22}}}\right) - \Phi\left(\frac{k_r-c}{\sqrt{\Delta_1^2\mathcal{V}_{11} +2\Delta_1\Delta_2\mathcal{V}_{12}+\Delta_2^2\mathcal{V}_{22}}}\right),
\end{align}
where $\Phi(\cdot)$ denotes the CDF of the standard normal distribution. Therefore, the overall probability of accepting a submitted lot with quality level $C_{py}=c$, referred to as the OC function and denoted by $\pi_A(c)$, is given by:
\begin{equation}
    \begin{aligned}
        \pi_A(c) = & P_a(c)+[P_S(c)\times P_a(c)] + \dots + [\{P_S(c)\}^{m-1}\times P_a(c)]\\
        = & \sum_{u=1}^m [\{P_S(c)\}^{u-1}\times P_a(c)]\\
        = & \frac{P_a(c)[1-\{P_S(c)\}^m]}{1-P_S(c)}.
    \end{aligned}
\end{equation}
As a special case, when $m=1$, the proposed sampling plan is equivalent to the $C_{
py}$-based the SSP, and its OC function reduces accordingly.

According to the proposed operating procedure, a submitted lot can be resampled at most $(m-1)$ times before a final decision is made. Therefore, the average sample number (ASN), defined as the expected number of sample units required to reach a lot acceptance or rejection decision, is an important measure of the plan's performance. For a submitted lot with quality level $C_{py}=c$, the ASN of the proposed sampling plan is given by:
\begin{equation}
    \begin{aligned}
        \text{ASN}(c) = & n+[n\times P_S(c)] + [n\times \{P_S(c)\}^2] + \dots + [n\times \{P_S(c)\}^{m-1}]\\
        = & \sum_{u=1}^m[n\times \{P_S(c)\}^{u-1} = \frac{n[1- \{P_S(c)\}^m]}{1-P_S(c)}.
    \end{aligned}
\end{equation}

Under Type-II HCS, the number of observed failures is the primary focus rather than the total sample size $n$. As the number of failures is random, its expected value is considered. Consequently, the ASN could be adapted to the situation of the Type-II hybrid censored sample by introducing the average failure numbers (AFN) as follows:
\begin{equation}
    \text{AFN}(c)=\frac{E[\mathcal{D}][1- \{P_S(c)\}^m]}{1-P_S(c)}.
\end{equation}

\begin{figure}[H]
    \centering
    \includegraphics[width=0.6\linewidth]{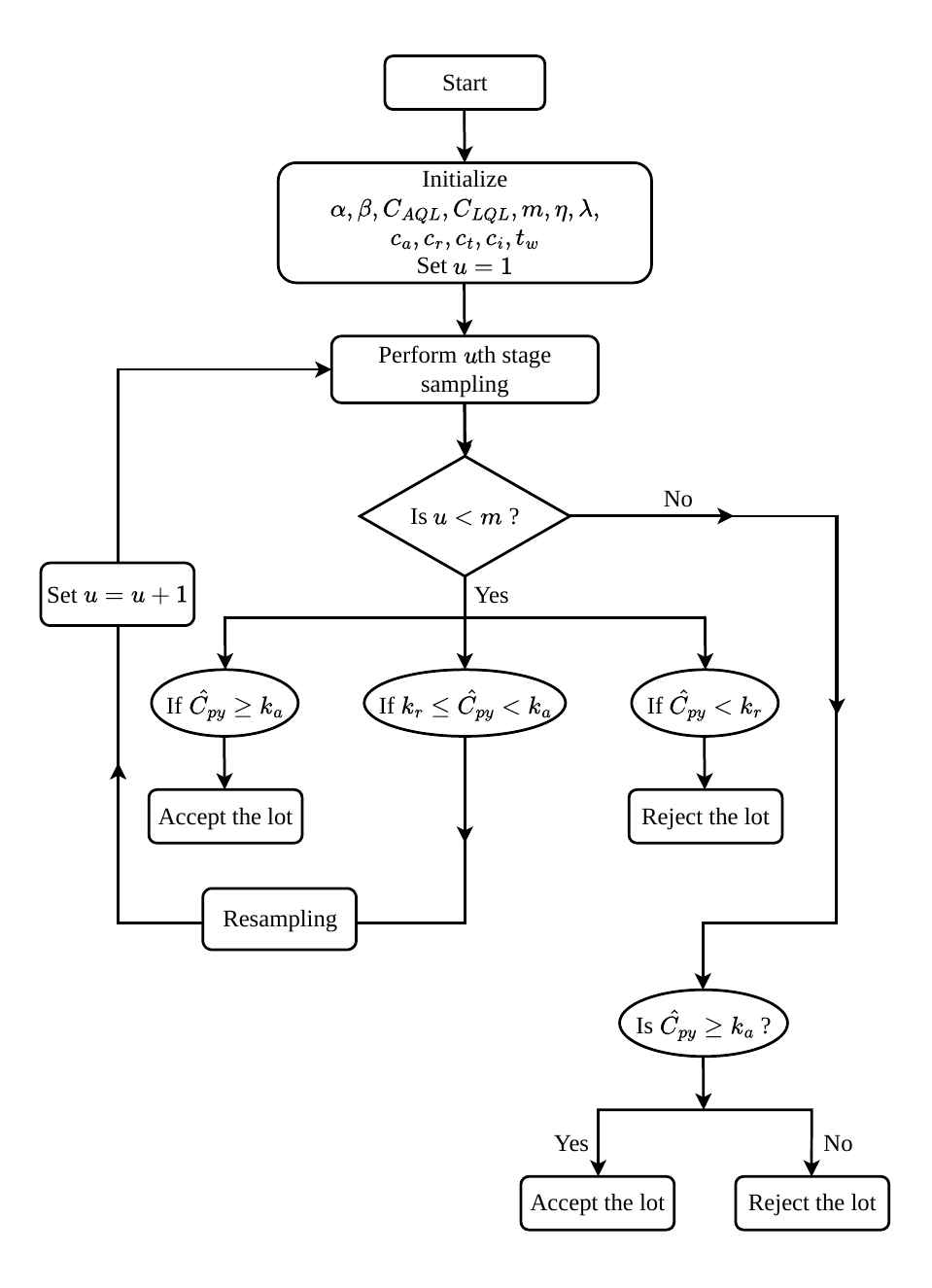}
    \caption{Flowchart of SIMSP}
    \label{simsp-flowchart}
\end{figure}

To satisfy the required quality levels $(C_{AQL}, C_{LQL})$ and the specified producer's and consumer's risks $(\alpha,\beta)$ while reducing inspection cost, the AFN is minimized. The corresponding optimization model is given by:
\begin{equation}\label{obj-1}
   \begin{aligned}
        \underset{(n,k_a,k_r)}{\text{Minimize}} & \quad  \frac{1}{2}E[\mathcal{D}]\left( \frac{[1- \{P_S(c_{AQL})\}^m]}{1-P_S(c_{AQL})} +  \frac{[1- \{P_S(c_{LQL})\}^m]}{1-P_S(c_{LQL})}\right) \\[10pt]
    \text{Subject to } & \quad \pi_A(c_{AQL}) = \frac{P_a(c_{AQL})[1-\{P_S(c_{AQL})\}^m]}{1-P_S(c_{AQL})} \geq                     1-\alpha\\
                       & \quad \pi_A(c_{LQL}) = \frac{P_a(c_{LQL})[1-\{P_S(c_{LQL})\}^m]}{1-P_S(c_{LQL})} \leq \beta\\
                       & \quad n\geq 2, ~ k_r< k_a.
   \end{aligned}
\end{equation}

The optimization problem in \eqref{obj-1} is intended to determine a feasible sampling plan that satisfies the prescribed producer's and consumer's risk requirements. In this section, the degree of censoring $(q)$ and the censoring time $(X_0)$ are treated as fixed design parameters, and the corresponding values of $n,k_a$, and $k_r$ can be obtained accordingly. However, this formulation does not account for the economic impact of time and cost. Therefore, we have extended the framework by incorporating time and cost into the objective function and formulated a new optimization problem that minimizes the expected total cost while jointly determining the optimal values of $n, r, X_0, k_a$, and $k_r$.

\subsection{Optimal sampling plans}
In many real industrial applications, especially in reliability and life-testing experiments, it is often impractical or time-consuming to observe all failure times. Instead, experiments are commonly conducted under censoring schemes, in which the test terminates after a predetermined number of failures have been observed. Under such conditions, the information contained in the censored sample must be incorporated into the design of the sampling plan. Ignoring censoring may lead to inefficient or unrealistic inspection strategies.

Another limitation of existing SIMSP based designs is that they primarily focus on minimizing the average number of failures as a single optimization objective. In practical decision-making environments, however, sampling plan design is rarely driven by a lone criterion. A more realistic formulation must account for several competing cost components, such as warranty loss, inspection expense, test duration cost, and the economic impact of accepting or rejecting a lot. Ignoring these factors may lead to plans that are statistically efficient but economically less desirable. Motivated by this observation, we replace the AFN based objective with a comprehensive cost function that reflects the operational burden of the inspection process. This allows the proposed design to balance statistical performance with economic efficiency, making the resulting plan more suitable for real industrial applications.

\subsubsection{Determining the cost function}

Previous studies identify four primary cost components associated with, or influenced by, a life testing scheme. These include: (i) the expense incurred when a lot is accepted, (ii) the expense resulting from rejecting a lot, (iii) the cost attributable to the duration of the test, and (iv) the inspection-related cost. Accordingly, the optimal life test plan is obtained by minimizing the overall cost function, denoted by $TC(n,r, X_0,k_a,k_r)$, which combines each of these cost elements into a single objective function.



In many practical situations, products are marketed under a pro-rata warranty policy. Under such a policy, the manufacturer’s warranty expenditure is directly influenced by the lot acceptance decision. Once a lot is accepted and released to market, product failures occurring within the warranty period may result in repair, replacement, compensation, servicing, and after-sales support costs. Consequently, acceptance of a comparatively low-reliability lot may result in substantial warranty-related expenses. Therefore, in the economic design of life test sampling plans, the expected warranty cost may be regarded as an important component of the lot acceptance cost.

\paragraph{Cost function:} Among various warranty policies, the pro-rata warranty policy \citep{blischke1992product, sen2022determination, chakrabarty2021optimum, chakrabarty2020optimum} is widely adopted in practice. Under this policy, the manufacturer's compensation decreases with the product's age at the time of failure. Hence, products that fail earlier in the warranty period receive higher compensation, whereas those that fail closer to the warranty expiration receive comparatively lower compensation. The mathematical formulation of the pro-rata warranty policy is expressed as follows:

\begin{equation}
    c_a^*=\begin{cases}
        c_a\left(1-\frac{t}{t_w}\right) & 0 \le t \leq t_w\\
        0 & t > t_w,
    \end{cases}
\end{equation}
where $t_w$ is the $w$th quantile of the product and $c_a$ is the cost of replacement. So, if the failure time is between $[0, t_w]$, the reimbursement under the pro-rata policy is determined by the remaining portion of the warranty period, which decreases over time. 

Hence, the expected warranty expenditure for each unit is obtained as
\begin{equation}
    w(\theta)=c_a\int_0^{t_w}\left(1-\frac{t}{t_w}\right)f_T(t)\,dt.
\end{equation}

Therefore, when $n$ units are selected from a lot of size $N$ for testing, the expected cost associated with warranty claims upon acceptance can be expressed as
\begin{equation}
    C_w=(N-n)w(\theta)\left(\frac{P_a(c)[1-\{P_S(c)\}^m]}{1-P_S(c)}\right).
\end{equation}
Note that the rejection cost is usually defined as the cost associated with units that are not tested. Let $c_r$ denote the per-unit cost associated with items that are excluded from testing. Then, the expected cost corresponding to the rejection of a lot can be expressed as
\begin{equation}
    C_r=(N-n)c_r\left(1-\frac{P_a(c)[1-\{P_S(c)\}^m]}{1-P_S(c)}\right).
\end{equation}
Now, if $c_t$ and $c_i$ are the cost per unit and the unit cost of inspection, the expected time consumption and the expected cost of inspection for failures are $ C_t = c_tE [\tau]$ and $ C_i = c_iE [D]$ respectively. Therefore, the aggregate cost function is 
\begin{equation*}
    \begin{aligned}
        TC(n,r,X_0,k_a,k_r) & = C_w + C_r + C_t + C_i \\
        & = (N-n)\left[c_r +  (w(\theta)-c_r)\left(\frac{P_a(c)[1-\{P_S(c)\}^m]}{1-P_S(c)}\right)\right] + c_tE[\tau] + c_iE[\mathcal{D}].
    \end{aligned}
\end{equation*}
Thus, our optimization problem in this case is formulated as
\begin{equation}\label{obj-2}
\begin{aligned}
\underset{(n,r,X_0,k_a,k_r)}{\text{Minimize}} \quad & TC(n,r,X_0,k_a,k_r)
\\[10pt]
\text{Subject to } \quad 
& \pi_A(c_{AQL}) = \frac{P_a(c_{AQL})\left[1-\{P_S(c_{AQL})\}^m\right]}{1-P_S(c_{AQL})} \geq 1-\alpha \\[6pt]
& \pi_A(c_{LQL}) = \frac{P_a(c_{LQL})\left[1-\{P_S(c_{LQL})\}^m\right]}{1-P_S(c_{LQL})} \leq \beta \\[6pt]
& n \geq 2, \quad r \geq 1, \quad k_r < k_a.
\end{aligned}
\end{equation}
However, it is quite difficult to solve the optimization problem formulated above in the Equation \eqref{obj-2} to determine the optimal plan, as it is a nonlinear mixed-integer programming problem. The difficulty mainly stems from the nonlinearity of the objective function, which consists of nonlinear terms and requires that the decision variables $n$ and $r$ be integers. To reduce the computational effort required to solve the problem, the normalized sample size $p_n=n/N$ is treated as a continuous decision variable. The integer sample size can be obtained through the floor function, $n=\lfloor p_nN\rfloor$. Likewise, instead of treating $r$ as a decision variable, the censoring proportion $q_r=1-r/n$ is considered, and the corresponding integer value of $r$ is obtained as $r=\lfloor(1-q)n\rfloor$. Since both $p_n$ and $q_r$ are continuous variables with $p_n,q_r\in(0,1)$, the original mixed-integer optimization problem is transformed into a nonlinear programming problem that can be efficiently solved using standard optimization techniques. To solve the optimization problem, we have used the \textit{nloptr} package in R version 4.5.3. Once the continuous optimum has been found, the integer equivalents of $n$ and $r$ can be obtained using the floor function. The entire optimization process is presented below as an algorithm \ref{algo-2}.
\begin{algorithm}[H]
\caption{Finding the optimal design}
\label{algo-2}
\KwIn{$\alpha,\beta,C_{AQL},C_{LQL},m,\eta,\lambda,N$}
\KwOut{$n^*,r^*,X_0^*,k_a^*,k_r^*,TC^*$}
Set the warranty period $t_w$ and all unit costs\\
Replace $n$ with $\lfloor p_nN\rfloor$ and $r$ with
$\lfloor(1-q_r)n\rfloor$ to transform the objective function
from $TC(n,r,X_0,k_a,k_r)$ to $TC(p_n,q_r,X_0,k_a,k_r)$\\
Solve the constrained optimization problem to obtain $(p_n^*,q_r^*, X_0^*,k_a^*,k_r^*,TC^*)$\\
Obtain $n^*=\lfloor p_n^*N\rfloor,~ r^*=\lfloor(1-q_r^*)n^*\rfloor$
to determine $(n^*,r^*)$\\
\end{algorithm}
\section{Analysis and comparison}\label{sec-5}
To evaluate the performance of the suggested acceptance sampling plan, a detailed numerical experiment is conducted under baseline parameter settings. The settings are selected to represent real-world industrial inspection problems and to enable meaningful comparison of various sampling plans.
The performance of the acceptance sampling plan is evaluated for four different producers' and consumers' risk settings: $(\alpha,\beta)=(0.01,0.01), (0.01,0.05), (0.05,0.01)$, and $(0.05,0.05)$. It should be noted that the risk values used in the above setting are typical in acceptance sampling applications and imply different levels of protection for both the producer and the consumer. The acceptable and limiting quality levels are set to $(C_{AQL}, C_{LQL})=(1.133,0.867), (1.125,0.875)$, and $(1.100,0.889)$ on the basis of the choice of $p_0=0.75,0.8$, and $0.90$ respectively. For example, the process characterized by $C_{py}=1.133$ is considered capable and ready to use for everyday production, while $C_{py}=0.867$ represents the minimum level of process capability. Thus, these quality levels define which lots are practically acceptable and which are unacceptable.

The lot size is fixed at $N=500$, denoting a moderate production lot size commonly observed in manufacturing setups. The lifetime of the product is modeled using the Chen distribution, with parameters $\eta=1.2$ and $\lambda=2.5$. In addition, the lower and upper specification limits are set to $L=0.1$ and $U=1$. The mathematical model accounts for inspection, testing, rejection, and warranty costs. Therefore, we obtain the expression of $TC$ considering the following unit costs: $c_a = 0.8$, $c_r = 1.4$, $c_t = 1.2$, and $c_i = 1.1$. For the given constant parameter values, the optimal plan parameters $(n^*,r^*, X_0^*, k_a, k_r)$ are determined by minimizing the model's total expected cost function $TC$, subject to constraints on producers' and consumers' risks. To facilitate the practical implementation of the proposed sampling plan, comprehensive design tables are generated for various combinations of $(C_{AQL}, C_{LQL}, \alpha, \beta)$ and presented for $m=2, 3$, and $4$. For instance, when $(C_{AQL}, C_{LQL}, \alpha, \beta)=(1.133, 0.867, 0.01, 0.05)$, the optimal design parameters are $(n^*, r^*, X_0^*, k_a, k_r)=(25, 24, 0.3014, 1.0094, 0.8692)$ for $m=2$ and $(24, 22, 0.4761, 1.0234, 0.8442)$ for $m=3$.

\begin{figure}[H]
    \centering
    \includegraphics[width=1\linewidth]{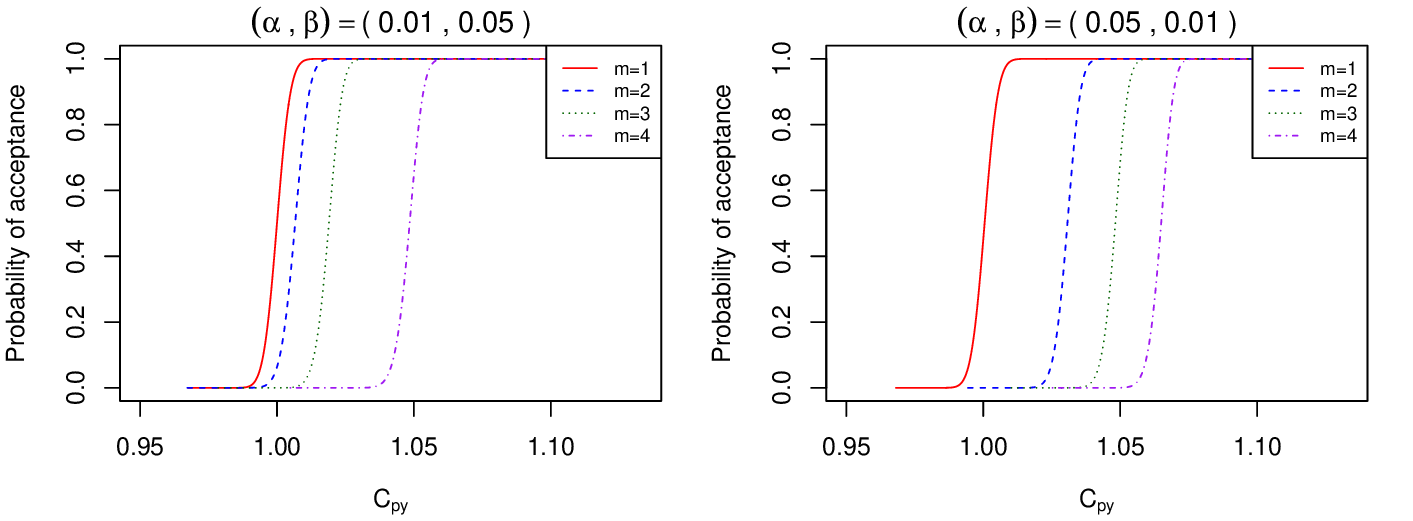}
    \caption{The OC curves for different resampling stages.}
    \label{oc-curves}
\end{figure}

The OC curve is an important measure for evaluating the discriminatory power of acceptance sampling plans. Figure \ref{oc-curves} compares the OC curves of the conventional variables SSP and the proposed resampling plan with $m=2, 3$, and $4$ under $(C_{AQL}, C_{LQL})=(1.133,0.867)$ for $(\alpha,\beta)=(0.01,0.05) $and $(\alpha,\beta)=(0.05,0.10)$. The results show that all plans satisfy the prescribed two-point conditions, thereby ensuring the desired protection of producers and consumers. Moreover, the proposed plans exhibit OC curves comparable to those of the SSP while achieving a lower expected total cost through the resampling strategy and warranty rebate mechanism, making them a more economical alternative.

The producers' risk $\alpha$ and consumers' risk $\beta$ significantly influence the discrimination capability of the sampling design used in acceptance sampling. The study findings indicate that as $\alpha$ and/or $\beta$ decreases, the required sample size and inspection costs increase, as more information will be needed to distinguish good from bad lots. It means that the expected inspection effort will increase as well. On the other hand, when higher risks are allowed, fewer samples will be sufficient for a better sampling process at a lower cost.

The effect of the quality levels $C_{AQL}$ and $C_{LQL}$ on the optimal design is also examined. The results show that, in most cases, the optimal sample size increases as the gap between $C_{AQL}$ and $C_{LQL}$ becomes smaller. This is because it becomes more difficult to distinguish between acceptable and unacceptable lots when the two quality levels are close, requiring more sample information to satisfy the specified risks. In contrast, a wider gap between $C_{AQL}$ and $C_{LQL}$ allows for easier lot classification and, therefore, smaller sample sizes. These findings suggest that overly narrow quality requirements can substantially increase the inspection effort and overall cost.

To evaluate the efficiency of the proposed sampling plan, its expected total cost is compared with that of the corresponding SSP under the same quality and risk requirements. The results show that both plans satisfy the required risk constraints, as shown in Figure \ref{cost-curves}; however, the proposed plan consistently achieves a lower expected total cost. This improvement is mainly due to the resampling strategy, which allows early acceptance or rejection of lots with clearly good or poor quality, thereby reducing unnecessary inspection effort. As a result, the proposed plan provides the same level of protection as the SSP while offering a more economical and efficient alternative for lot disposition. 

\begin{figure}[H]
    \centering
    \includegraphics[width=1\linewidth]{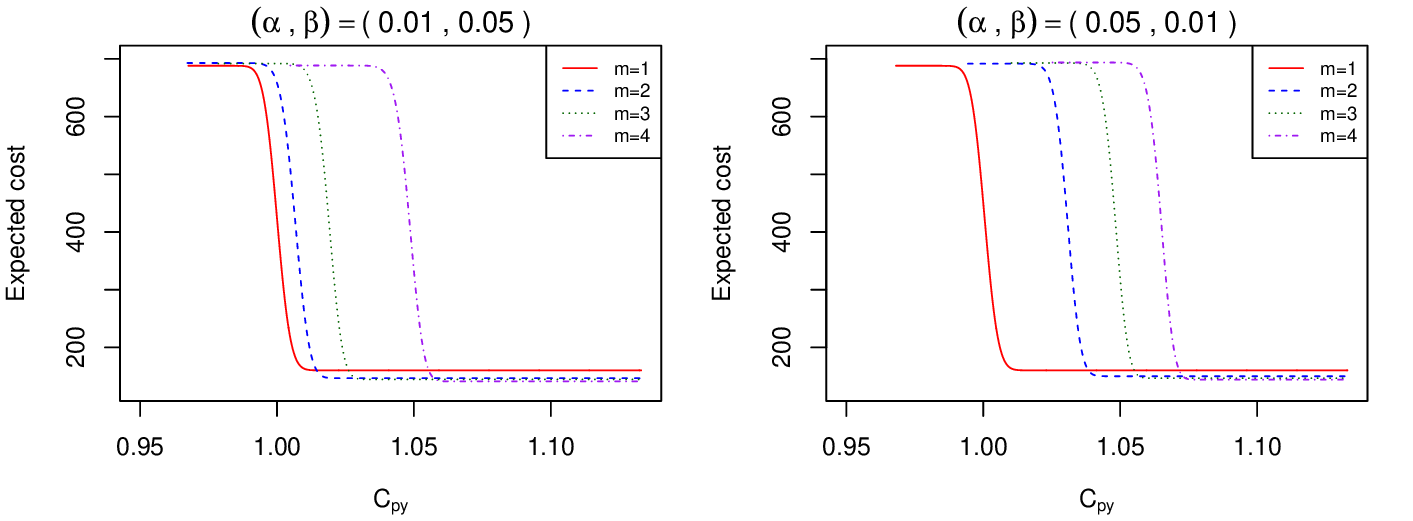}
    \caption{The expected cost for different resampling stages.}
    \label{cost-curves}
\end{figure}

Finally, to assess the effectiveness of the proposed model for small sample sizes while incorporating the producer's risk $(\alpha)$, an extensive Monte Carlo simulation study is performed. Because the lot acceptance criterion is established using the asymptotic distribution, the simulation analysis serves to examine whether the proposed sampling plan remains reliable when the sample size is limited. In determining the optimal designs, for different combinations of $C_{AQL},C_{LQL}$, $m$ with $(\alpha,\beta)=(0.01,0.05)$, the maximum likelihood estimates of $C_{py}$ are obtained for each simulated sample. The proportion of rejected lots is then compared with the nominal producer's risk $\alpha$. As shown in Table \ref{monte-carlo-est-alpha}, the estimated $\hat{\alpha}$ are satisfied with agreement with the specified risk levels (i.e. $\hat{\alpha}\leq\alpha$), indicating that the proposed methodology performs well even for relatively small sample sizes. The simulation procedure is in the Algorithm \ref{algo-3}.

\begin{table}[H]
\small
\centering
\caption{The optimal plan parameters for $m = 2$.}
\setlength{\tabcolsep}{12pt}
\resizebox{\textwidth}{!}{
\begin{tabular}{ccccccccccc}
\toprule
$p_0$ & $C_{AQL}$ & $C_{LQL}$ & $\alpha$ & $\beta$ & $n^*$ & $r^*$ & $X_0^*$ & $k_a$ & $k_r$ & $TC$\\
\midrule
0.75 & 1.133 & 0.867& 0.01 & 0.01 & 30 & 29 &0.4568& 1.0304 &0.9089& 125.89\\
&&&&                         0.05 & 25 &  24 &0.3014 &1.0094& 0.8692& 118.95\\
\cmidrule(l){4-11}
&&&0.05 & 0.01 & 29  & 28 &0.3497& 1.0333 &0.9100 &125.81 \\
& & &&         0.05 & 22 &  21& 0.2851& 1.0121 &0.8758& 118.81\\
\midrule
0.80 & 1.125& 0.875& 0.01 & 0.01 & 36   &32 &0.4912 &1.0366 &0.8122 &127.08\\
&&&&           0.05 & 33  & 26 &0.4775 &1.0192& 0.8436& 123.27\\
\cmidrule(l){4-11}
&&&0.05 & 0.01 & 30 &  28 &0.5728 &1.0337& 0.8705 &128.57\\
& &  &&        0.05 & 29  & 25 &0.5369& 1.0173 &0.8795& 120.06\\
\midrule
0.90 & 1.100& 0.889& 0.01 & 0.01 & 42 &  38 &0.2687 &1.0142 &0.9132 &132.08\\
&& &&          0.05 & 26  & 24 &0.4574 &1.0064 &0.8823 &121.53\\
\cmidrule(l){4-11}
&&&0.05 & 0.01 & 30  & 29 &0.4033 &1.0251& 0.9246 &128.58\\
& &  &&        0.05 & 25  & 23 &0.4776 &1.0073& 0.8968 &121.53\\
\bottomrule
\end{tabular}
}
\end{table}

\begin{table}[H]
\small
\centering
\caption{The optimal plan parameters for $m = 3$.}
\setlength{\tabcolsep}{12pt}
\resizebox{\textwidth}{!}{
\begin{tabular}{ccccccccccc}
\toprule
$p_0$ & $C_{AQL}$ & $C_{LQL}$ & $\alpha$ & $\beta$ & $n^*$ & $r^*$ & $X_0^*$ & $k_a$ & $k_r$ & $TC$\\
\midrule
0.75 & 1.133 & 0.867& 0.01 & 0.01 & 29 & 28 &0.5002 &1.0389 &0.8819 &122.79 \\
&&&&           0.05 & 24 & 22& 0.4761 &1.0234 &0.8442& 117.06\\
\cmidrule(l){4-11}
&&&0.05 & 0.01 & 25&   24 &0.4134 &1.0524& 0.8887& 120.63\\
& & &&         0.05 & 24  & 20& 0.3642 &1.0401 &0.8507 &116.81\\
\midrule
0.80 & 1.125& 0.875& 0.01 & 0.01 & 31 & 30 &0.5587 &1.0309& 0.8857& 124.36\\
&&&&           0.05 & 26 &  21 &0.4716 &1.0370& 0.8594& 117.33\\
\cmidrule(l){4-11}
& & & 0.05 & 0.01 & 29 &  27 &0.4717 &1.0541 &0.8917 &122.62\\
& &  &&        0.05 & 22 & 20 &0.5807 &1.0688 &0.8595 &115.83\\
\midrule
0.90 & 1.100& 0.889& 0.01 & 0.01 & 33 &  29 &0.3081& 1.0381 &0.9071 &124.92\\
&& &&          0.05 & 24  & 23& 0.3634 &1.0042 &0.8760 &120.69\\
\cmidrule(l){4-11}
&&&0.05 & 0.01 & 28  & 26 &0.4427 &1.0386 &0.9208 &123.56\\
& &  &&        0.05 & 22 &  20 &0.4365 &1.0250 &0.8956& 116.98\\
\bottomrule
\end{tabular}
}
\end{table}

\begin{table}[H]
\small
\centering
\caption{The optimal plan parameters for $m = 4$.}
\setlength{\tabcolsep}{12pt}
\resizebox{\textwidth}{!}{
\begin{tabular}{ccccccccccc}
\toprule
$p_0$ & $C_{AQL}$ & $C_{LQL}$ & $\alpha$ & $\beta$ & $n^*$ & $r^*$ & $X_0^*$ & $k_a$ & $k_r$ & $TC$\\
\midrule
0.75 & 1.133 & 0.867& 0.01 & 0.01 & 25  & 24 &0.3902 &1.0801 &0.8623 &117.99\\
&&&&           0.05 & 24 &  19 &0.4835 &1.0544& 0.8510& 115.25\\
\cmidrule(l){4-11}
&&&0.05 & 0.01 & 22 &  21 &0.4184 &1.0704 &0.8724 &117.66\\
& & &&         0.05 & 22  & 18 &0.3940 &1.0551 &0.8610 &114.67\\
\midrule
0.80 & 1.125& 0.875& 0.01 & 0.01 & 27 &   24 & 0.4301 & 1.0657 & 0.8624 & 119.43\\
&&&&           0.05 & 22  &  18 & 0.5550 & 1.0562 & 0.8397&  114.40\\
\cmidrule(l){4-11}
&&&0.05 & 0.01 & 26 & 24 &0.6361 &1.0507 &0.8715& 118.89\\
& &  &&        0.05 & 19 &   16 & 0.5626 & 1.0624&  0.8417 & 114.07\\
\midrule
0.90 & 1.100& 0.889& 0.01 & 0.01 & 30  & 26 &0.2137 &1.0561& 0.8717 &122.01\\
&& &&          0.05 & 19 &  17 &0.4161 &1.0445& 0.8601 &114.55\\
\cmidrule(l){4-11}
&&&0.05 & 0.01 & 26 &  24 &0.4339 &1.0510 &0.8936 &120.25\\
& &  &&        0.05 & 18 &  16 &0.4262 &1.0500& 0.8636& 113.11 \\
\bottomrule
\end{tabular}
}
\end{table}

\begin{table}[H]
    \small
    \centering
    \caption{Estimated $\hat{\alpha}$ under $(\alpha,\beta)=(0.01,0.05)$}
    \label{monte-carlo-est-alpha}
    \setlength{\tabcolsep}{12pt}
    \resizebox{\textwidth}{!}{
    \begin{tabular}{ccccccccccc}
        \toprule
        $m$ & $p_0$ & $C_{AQL}$ & $C_{LQL}$ & $n^*$ & $r^*$ & $X_0^*$ & $k_a$ & $k_r$ & $TC^*$ & $\hat{\alpha}$\\
        \midrule
        2 & 0.75 & 1.133 & 0.867 & 25 & 24 & 0.3014 & 1.0094 & 0.8692 & 118.95 & 0.0041\\
          & 0.80 & 1.125 & 0.875 & 33 & 26 & 0.4775 & 1.0192 & 0.8436 & 123.27 & 0.0033\\ 
          & 0.90 & 1.100 & 0.889 & 26 & 24 & 0.4574 & 1.0064 & 0.8823 & 121.53 & 0.0019\\
        3 & 0.75 & 1.133 & 0.867 & 24 & 22 & 0.4761 & 1.0234 & 0.8442 & 117.06 & 0.0023\\
          & 0.80 & 1.125 & 0.875 & 26 & 21 & 0.4716 & 1.0370 & 0.8594 & 117.33 & 0.0018\\ 
          & 0.90 & 1.100 & 0.889 & 24 & 23 & 0.3634 & 1.0042 & 0.8760 & 120.69 & 0.0010\\
        4 & 0.75 & 1.133 & 0.867 & 24 & 19 & 0.4835 & 1.0544 & 0.8510 & 115.25 & 0.0017\\
          & 0.80 & 1.125 & 0.875 & 22 & 18 & 0.5550 & 1.0562 & 0.8397 & 114.40 & 0.0013\\
          & 0.90 & 1.100 & 0.889 & 19 & 17 & 0.4161 & 1.0445 & 0.8601 & 114.55 & 0.0004\\
        \bottomrule
    \end{tabular}
    }
\end{table}
\begin{algorithm}[H]
\caption{Monte Carlo simulation for $\hat{\alpha}$}
\label{algo-3}
\KwIn{$\alpha,\beta,C_{AQL},C_{LQL},m,\lambda,n^*,r^*,X_0^*,k_a^*,k_r^*$}
\KwOut{$\hat{\alpha}$}
Determine $\eta_{\mathrm{AQL}}$ such that
$C_{py}(\eta_{\mathrm{AQL}},\lambda)=C_{\mathrm{AQL}}$\\
Initialize: count=0\\
\For{$i=1$ \KwTo $10,000$}{
    \For{$u=1$ \KwTo $m$}{
       Generate HC-II sample from $Chen(\eta_{\mathrm{AQL}},\lambda)$ using $n^*,r^*,X_0^*$\\
        Compute $\hat{C}_{py}=C_{py}(\hat{\eta},\hat{\lambda})$
        \If{$u<m$}{
            \uIf{$\hat{C}_{py}\ge k_a$}{
                count=count+1\\
                \textbf{break}
            }
            \uElseIf{$\hat{C}_{py}< k_r$}{
                \textbf{break}
            }
            \uElseIf{$k_r \leq \hat{C}_{py} < k_a$}{
            Continue to the next stage
            }
        }
        \Else{
            \uIf{$\hat{C}_{py}\ge k_a$}{
                count=count+1
            }
        }
    }
}
Compute $\hat{\alpha}=1-\frac{\text{count}}{10,000}$
\end{algorithm}
\section{Case study}\label{sec-6}
To assess the performance and practical applicability of the proposed optimal sampling plan, a real-world dataset has been analyzed. The following dataset from \citep{lawless2011statistical} is used to illustrate our sampling plan. The data consist of the recorded failure mileages (in thousands of miles) for a collection of locomotive controls obtained from a reliability study. We have divided the whole dataset by 100 for computational purposes and presented it in the Table \ref{real-data}. 
\begin{table}[H]
    \centering
    \caption{The failure times for the 37 failed units}
    \resizebox{\textwidth}{!}{
    \begin{tabular}{cccccccccccccc}
    \toprule
        0.225 & 0.375 & 0.460 & 0.485 & 0.515 & 0.530 & 0.545 & 0.575 & 0.665 & 0.680 & 0.695 & 0.765 & 0.770 & 0.785\\
        0.800 & 0.815 & 0.820 & 0.830 & 0.840 & 0.915 & 0.935 & 1.025 & 1.070 & 1.085 & 1.125 & 1.135 & 1.160 & 1.170\\
        1.185 & 1.190 & 1.200 & 1.225 & 1.230 & 1.275 & 1.310 & 1.325 & 1.340\\
        \bottomrule
    \end{tabular}
    }
    \label{real-data}
\end{table}
Furthermore, the Kolmogorov-Smirnov goodness-of-fit test was performed to determine whether our model is suitable for analyzing the given dataset. The computed test statistics, along with their $p$-values, are 0.13478 and 0.4717. This suggests that the Chen distribution provides a fit for considering the dataset. To better understand the data distribution and the theoretical vs. empirical CDFs, a P-P plot and a histogram of the data with fitted PDFs are presented in Figure \ref{fit} to demonstrate the fit to the dataset. Now, considering the data from the Chen distribution and accounting for censoring, the distribution parameters are estimated by maximum likelihood to be $\hat\eta=0.1323$ and $\hat\lambda=1.4663$, respectively.
\begin{figure}[H]
    \centering
    \includegraphics[width=1\linewidth]{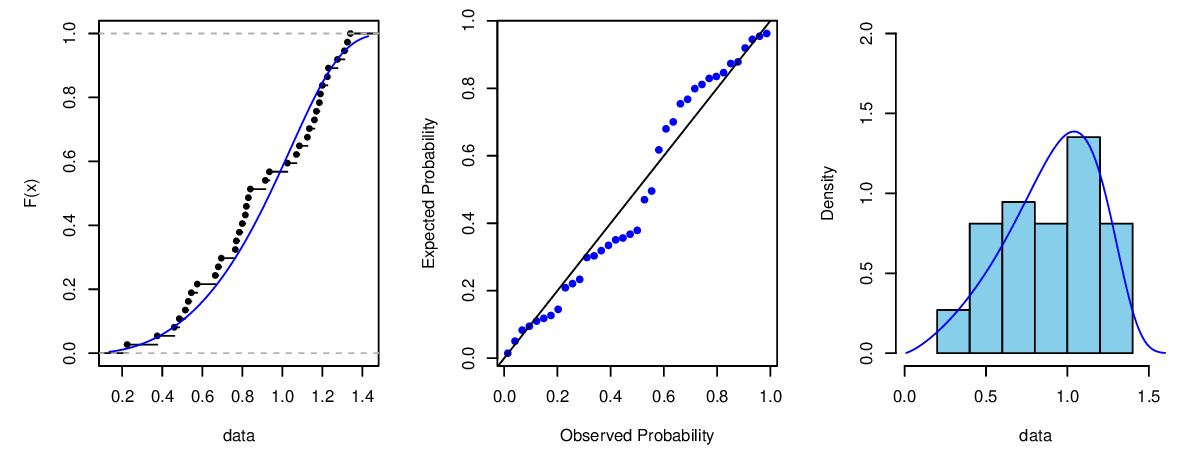}
    \caption{Empirical CDF with theoretical CDF, P-P, and histogram with fitted PDF plot.}
    \label{fit}
\end{figure}
For the numerical illustration, the specification limits are fixed at $(LSL, USL)=(0.1, 2)$. Suppose the acceptable and limiting quality levels are specified as $(c_{AQL}, c_{LQL})=(1.133,0.867)$, with the producer's and consumer's risks set to $(\alpha,\beta)=(0.01,0.05)$, respectively. The maximum allowable number of sampling stages is $m=3$. Under these design requirements, the lot acceptance probability must be at least $0.99$ when the submitted lot quality is $C_{py}=1.133$. In contrast, it must not exceed $0.05$ when the lot quality deteriorates to $C_{py}=0.867$. Assuming the true parameter values are $(\eta,\lambda)=(0.1323,1.4663)$, and solving the optimization problem in Equation \eqref{obj-2} yields the optimal plan parameters $(n,r,X_0,k_a,k_r)=(21,20,0.5059,1.0315,0.8477)$ with the corresponding minimum total cost is $TC=119.28$.

Based on the computed sampling plan, $(n, r, X_0) = (21, 20, 0.5059)$, we take a sample from the full data set in Table \ref{real-data} to mimic a lot inspection scenario as presented in the Table \ref{real-lot}. Then, based on the collected units, we have estimated the MLEs $\hat{\eta}$ and $\hat{\lambda}$. Figure \ref{profile-lambda} represents the profile log-likelihood of $\lambda$. Next, we have calculated the sample estimator as $\hat{C}_{py}=1.3203$, which exceeds the critical value for acceptance, $k_a=1.0315$. Hence, the entire submitted lot can be accepted. On the other hand, if $\hat{C}_{py} < 0.8477$, the lot should be rejected immediately, and if $\hat{C}_{py}$ falls within the resampling region $[0.8477,1.0315)$, a second stage sample is to be drawn from the submitted lot. The decision regarding the lot at this stage follows the same acceptance criterion applied during the first stage. If inspection proceeds to the third and final stage, the lot is accepted only if $\hat{C}_{py} \geq 1.0315$; otherwise, the lot is rejected.

\begin{table}[H]
    \centering
    \caption{Collected 28 failure times}
    \resizebox{\textwidth}{!}{
    \begin{tabular}{cccccccccccccc}
    \toprule
        0.225 & 0.375 & 0.460 & 0.485 & 0.515 & 0.530 & 0.545 & 0.665 & 0.695 & 0.765 & 0.770 & 0.815 & 0.820 & 0.830\\
        1.070 & 1.125 & 1.170 & 1.200 & 1.230 & 1.310\\
        \bottomrule
    \end{tabular}
    }
    \label{real-lot}
\end{table}
\begin{figure}[H]
    \centering
    \includegraphics[width=0.45\linewidth]{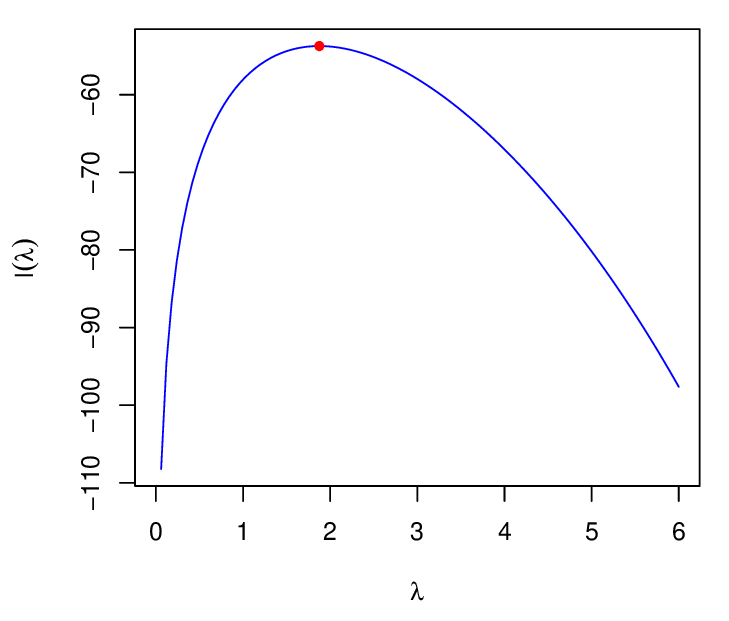}
    \caption{Profile log-likelihood plot of $\lambda$ for real data}
    \label{profile-lambda}
\end{figure}
\subsection{Sensitivity analysis}
The construction of the proposed SIMSP revealed that the parameters of the Chen distribution of $T$ play a crucial role in determining the optimal sampling design. Consequently, a comprehensive sensitivity analysis is conducted to assess the effect of parameter misspecification on the performance, robustness, and reliability of the proposed sampling plan. The above dataset, as presented in the Table \ref{real-data}, is used to conduct the sensitivity analysis. The estimates (respective standard errors) of $\eta$ and $\lambda$ are found to be $0.1323 (0.00068)$ and $1.4663 (0.04732)$, respectively. For sensitivity analysis, using the estimates and their respective standard errors, three sets of values (estimate, estimate + standard error, estimate - standard error) are computed for each parameter. Using three pairs of $(c_{AQL},c_{LQL})$ and $(\alpha,\beta)=(0.01,0.05)$, the optimal results corresponding to each of the nine sets of parameters are obtained with maximum allowable sampling stages $m=3$.

The results presented in Table \ref{sensitivity-1} prove the robustness of the SIMSP to moderate misspecifications of the Chen distribution parameters. Deviations of $\eta$ and $\lambda$ by one standard error from their estimated values result in minor changes in the design parameters and the associated total cost. Both the sample size and the acceptance number are stable for any possible combination of the parameters. Although the impact of $\lambda$ on the design seems to be somewhat stronger than the effect of $\eta$, the differences in terms of both cost and design parameters are quite minor. Thus, the analysis proves the applicability and robustness of the proposed sampling scheme.

The sensitivity of the SIMSP to the lot size $N$ is shown in Table \ref{sensitivity-2}. It can be seen that the total cost $TC$ of the process increases sharply as the lot size increases from $250$ to $1000$, which is quite natural, given the increased costs of inspection, testing, and replacement for larger lots. On the other hand, there is a gradual increase in the optimal sample size $n$ and optimal acceptance number $r$, implying that the effort required for the sampling process increases very slowly relative to the lot size. Furthermore, the value of the optimal termination time $X_0$ is fairly stable across the large-scale production environments. Furthermore, as warranty cost represents a significant component of the proposed cost function, its impact on the optimal design is examined. Figure \ref{opt-vs-plan} shows that the inclusion of warranty cost results in a decrease in the optimal censoring time and an increase in the optimal sample size. It is also evident that the warranty period substantially affects the optimal total cost, with longer warranty durations leading to higher optimal costs.
\begin{figure}[H]
    \centering
    \includegraphics[width=0.9\linewidth]{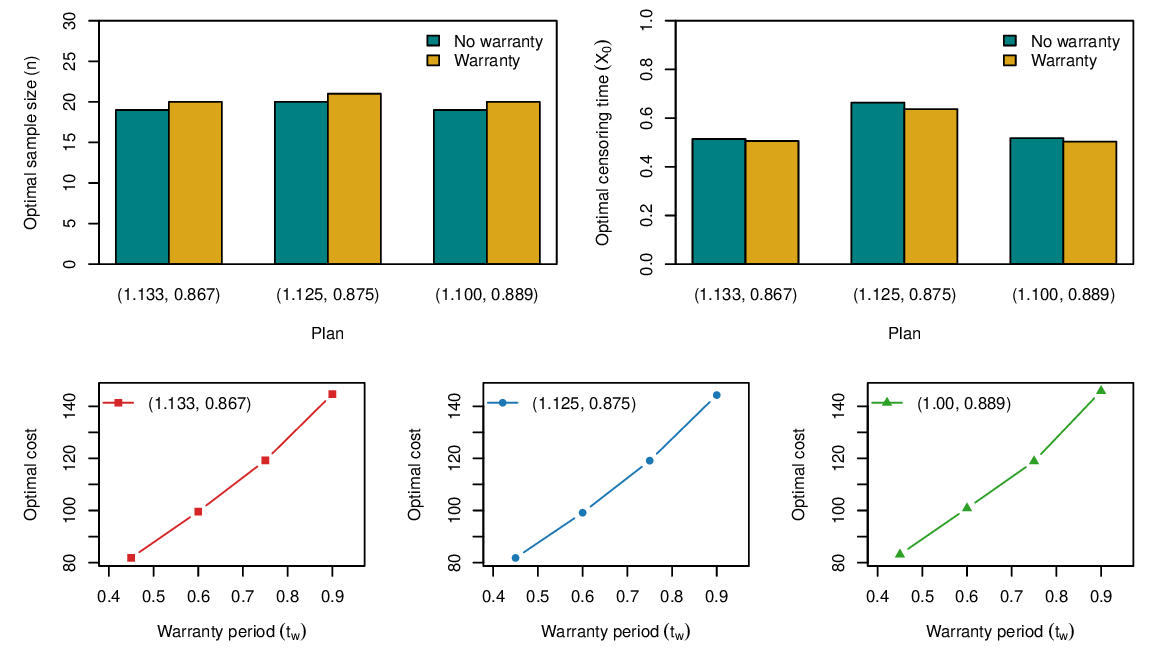}
    \caption{Influence of warranty cost and period on the optimal design.}
    \label{opt-vs-plan}
\end{figure}
\begin{table}[H]
    \scriptsize
    \centering
    \caption{Sensitivity analysis of failure data of locomotive controls.}
    \label{sensitivity-1}
    \setlength{\tabcolsep}{13pt}
    \resizebox{\textwidth}{!}{
    \begin{tabular}{ccccccccc}
        \toprule
        $P_0$ & $c_{AQL}$ & $c_{LQL}$ & $\hat{\eta}$ & $\hat{\lambda}$ & $n^*$ & $r^*$ & $X_0^*$ & $TC$\\
        \cmidrule{1-9}
        0.75 & 1.133 & 0.867 & 0.13162 & 1.41898 & 23 & 22 & 0.7218 & 123.29\\
        & & &                  0.13162 & 1.46630 & 22 & 21 & 0.5098 & 119.70\\
        & & &                  0.13162 & 1.51362 & 20 & 19 & 0.5599 & 115.65\\
        & & &                  0.13230 & 1.41898 & 23 & 22 & 0.7426 & 123.31\\
        & & &                  0.13230 & 1.46630 & 21 & 20 & 0.5059 & 119.22\\
        & & &                  0.13230 & 1.51362 & 18 & 17 & 0.5016 & 114.81\\
        & & &                  0.13298 & 1.41898 & 22 & 21 & 0.6087 & 122.95\\
        & & &                  0.13298 & 1.46630 & 23 & 22 & 0.5155 & 120.49\\
        & & &                  0.13298 & 1.51362 & 19 & 18 & 0.5014 & 115.18\\
        \cmidrule{1-9}
        0.80 & 1.125 & 0.875 & 0.13162 & 1.41898 & 24 & 23 & 0.6050 & 123.76\\
        & & &                  0.13162 & 1.46630 & 19 & 18 & 0.5596 & 119.00\\
        & & &                  0.13162 & 1.51362 & 17 & 16 & 0.5166 & 114.94\\
        & & &                  0.13230 & 1.41898 & 22 & 21 & 0.5727 & 123.07\\
        & & &                  0.13230 & 1.46630 & 20 & 19 & 0.6368 & 119.13\\
        & & &                  0.13230 & 1.51362 & 17 & 16 & 0.5051 & 114.91\\
        & & &                  0.13298 & 1.41898 & 21 & 20 & 0.6901 & 122.99\\
        & & &                  0.13298 & 1.46630 & 19 & 18 & 0.5021 & 118.82\\
        & & &                  0.13298 & 1.51362 & 21 & 20 & 0.5004 & 116.48\\
        \cmidrule{1-9}
        0.90 & 1.100 & 0.889 & 0.13162 & 1.41898 & 24 & 23 & 0.6654 & 125.56\\
        & & &                  0.13162 & 1.46630 & 21 & 20 & 0.5295 & 120.40\\
        & & &                  0.13162 & 1.51362 & 20 & 19 & 0.5105 & 116.95\\    
        & & &                  0.13230 & 1.41898 & 24 & 23 & 0.5999 & 125.65\\
        & & &                  0.13230 & 1.46630 & 20 & 19 & 0.5084 & 120.43\\
        & & &                  0.13230 & 1.51362 & 20 & 19 & 0.5140 & 116.81\\    
        & & &                  0.13298 & 1.41898 & 24 & 23 & 0.7088 & 125.69\\
        & & &                  0.13298 & 1.46630 & 20 & 19 & 0.5039 & 120.39\\
        & & &                  0.13298 & 1.51362 & 21 & 20 & 0.5473 & 117.19\\
        \bottomrule
    \end{tabular}
    }
\end{table}
\begin{table}[H]
    \scriptsize
    \centering
    \caption{Optimal design for different lot sizes with $(\eta,\lambda)=(0.1323,1.4663)$.}
    \label{sensitivity-2}
    \setlength{\tabcolsep}{18pt}
    \resizebox{\textwidth}{!}{
    \begin{tabular}{cccccccc}
        \toprule
        $P_0$ & $c_{AQL}$ & $c_{LQL}$ & $N$ & $n^*$ & $r^*$ & $X_0^*$ & $TC$\\
        \cmidrule{1-8}
        0.75 & 1.133 & 0.867 & 250 & 18 & 17 & 0.5483 & 68.38\\
        & & &                  500 & 21 & 20 & 0.5059 & 119.22\\
        & & &                  750 & 22 & 21 & 0.5199 & 169.08\\
        & & &                 1000 & 25 & 24 & 0.5980 & 219.70\\
        \cmidrule{1-8}
        0.80 & 1.125 & 0.875 & 250 & 18 & 17 & 0.5032 & 68.51\\
        & & &                  500 & 20 & 19 & 0.6368 & 119.13\\
        & & &                  750 & 22 & 21 & 0.5500 & 168.92\\
        & & &                 1000 & 23 & 22 & 0.6473 & 218.35\\
        \cmidrule{1-8}
        0.90 & 1.100 & 0.889 & 250 & 19 & 18 & 0.5980 & 69.06\\
        & & &                  500 & 19 & 18 & 0.5037 & 118.86\\
        & & &                  750 & 22 & 21 & 0.5385 & 169.62\\
        & & &                 1000 & 23 & 22 & 0.5114 & 219.04\\        
        \bottomrule
    \end{tabular}
    }
\end{table}
\section{Conclusion and Future Work}\label{sec-7}
In this paper, a stage-independent multiple sampling plan based on the generalized process capability index $C_{py}$ is proposed for lot sentencing under a Type-II HCS while incorporating a pro-rata warranty rebate policy. Unlike the conventional multiple sampling plan, the proposed plan assumes that the inspection decision at each sampling stage is independent of the outcomes from previous stages, which substantially simplifies the derivation of the OC function while preserving the economic advantages of multistage inspection. The plan parameters are obtained by solving a constrained optimization model that minimizes the expected total cost, while satisfying the specified producer and consumer risk requirements. Extensive simulation studies illustrate that the proposed SIMSP significantly reduces the average sample size and total expected cost compared to traditional one-stage sampling, while still providing sufficient guarantees for the producer's and consumer's protection levels. Finally, a real data application demonstrates the practical implementation of the proposed approach. The sensitivity analysis further helps to study the influence of parameters on optimal design. Due to the flexibility of the generalized process capability index, the suggested framework can be easily modified to suit different situations.

Several directions for future research remain, including extending the proposed framework to other censoring schemes, such as progressive hybrid and adaptive hybrid censoring, as well as Bayesian and bootstrap estimation methods. It can also be implemented for accelerated life testing, competing risks, and dependent lifetime distributions. Furthermore, incorporating more sophisticated warranty policies, such as combined free replacement/pro-rata warranties or renewable warranties, along with a variable cost structure, will make the developed sampling plan even more realistic.

 



\section*{Conflict of interest statement} There is no conflict of interest among the authors of this manuscript.
\section*{Data availability statement} The authors confirm that the data that support the findings of this study are available in the literature. 
\section*{Funding} This work of Yogesh Mani Tripathi is partially financially supported by a grant MTR/2022/000183 from the Science and Engineering Research Board, India. The research work of Tanmay Kayal is partially supported by the Indian Institute of Technology Mandi, India, with the SEED project grant (Project No. IITM/SG/TK/141).

\bibliographystyle{unsrt}
\bibliography{references}

\section*{Appendix A. Proof of the Theorem \ref{th-1}}\label{appendix-A}
By taking the derivative of Equation \eqref{log-lik} to $\eta$ and setting it to zero, then \eqref{eta-eq} can be obtained simply. Then, we will demonstrate that the log-likelihood function \eqref{log-lik} will reach its maximal value at $\hat{\eta}$.

Let, $t=\frac{\eta}{\hat{\eta}}$, then using $\ln{t}\leq t-1$, we have
\[ \ln{\eta} \leq -\eta\frac{\sum_{i=1}^d\left(1-e^{x_{i:n}^\lambda}\right) + (n-d)\left(1-e^{\tau_0^\lambda}\right)}{d} -1 + \ln{\hat{\eta}}.\]

Hence, the log-likelihood function can be obtained as
\[\ell(\eta,\lambda) \leq d\ln{\hat{\eta}} + d\ln{\lambda} + \lambda\sum_{i=1}^d\ln{x_{i:n}} + \sum_{i=1}^dx_{i:n}^\lambda -d.\]

Now, using the Equation \eqref{lam-eq} with replacing $d$, we obtain the log-likelihood function as
\begin{align*}
    \ell(\eta,\lambda) & \leq d\ln{\hat{\eta}} + d\ln{\lambda} + \lambda\sum_{i=1}^d\ln{x_{i:n}} + \sum_{i=1}^dx_{i:n}^\lambda + \hat{\eta}\sum_{i=1}^d\left(1-e^{x_{i:n}^\lambda}\right) + (n-d)\hat{\eta}\left(1-e^{\tau_{0}^\lambda}\right)\\
    & = \ell(\hat{\eta},\lambda).
\end{align*}
This completes the proof.
\section*{Appendix B. Proof of the Theorem \ref{th-2}}\label{appendix-B}
By taking the derivative of the profile log-likelihood function \eqref{lam-eq} with respect to $\lambda$ and setting it to zero, we can obtain the maximum likelihood estimate of $\lambda$ by solving $\psi(\lambda)=0$ where $\psi(\lambda)$ is given below. Now
the existence and uniqueness of the MLE of $\lambda$ are established. We have
\[\psi(\lambda)=\frac{d}{\lambda}+\sum_{i=1}^d\left(1+x_{i:n}^\lambda\right)\ln{x_{i:n}}-d\frac{\phi^{'}(\lambda)}{\phi(\lambda)},\]
where $$\phi^{'}(\lambda)=-\sum_{i=1}^de^{x_{i:n}^\lambda}x_{i:n}^\lambda\ln{x_{i:n}}-(n-d)e^{\tau_0^\lambda}\tau_0^\lambda\ln{\tau_0}$$ and
$$\phi(\lambda)=\sum_{i=1}^d\left(1-e^{x_{i:n}^\lambda}\right) + (n-d)\left(1-e^{\tau_0^\lambda}\right).$$
It can be easily seen that when $\lambda\to 0$, $\lim_{\lambda\to 0}\frac{d}{\lambda}\to+\infty$ and other terms are constant, so $\lim_{\lambda\to 0}\psi(\lambda)\to+\infty$. When $\lambda\to+\infty$, $\lim_{\lambda\to +\infty}\frac{d}{\lambda}\to0$, then $\psi(\lambda)$ becomes $$\psi_1(\lambda)=\sum_{i=1}^d\left(1+x_{i:n}^\lambda\right)\ln{x_{i:n}}-d\frac{\phi^{'}(\lambda)}{\phi(\lambda)}.$$ Then limit at $\lambda\to+\infty$ is evaluated using the following limits:
\begin{enumerate}[(i)]
    \item When $0<x_{i:n}<1$, $\lim_{\lambda\to+\infty} x_{i:n}^\lambda=0, \ln{x_{i:n}}<0$, and $\lim_{\lambda\to+\infty}e^{x_{i:n}^\lambda}=1$
    \item When $x_{i:n}=1$, $\lim_{\lambda\to+\infty} x_{i:n}^\lambda=1, \ln{x_{i:n}}=0$, and $\lim_{\lambda\to+\infty}e^{x_{i:n}^\lambda}=e$
    \item When $x_{i:n}>1$, $\lim_{\lambda\to+\infty} x_{i:n}^\lambda=+\infty, \ln{x_{i:n}}>0$, and $\lim_{\lambda\to+\infty}e^{x_{i:n}^\lambda}=+\infty.$
\end{enumerate}
Then we obtain $\lim_{\lambda\to+\infty}\psi_1(\lambda)=\text{constant}<0$, and consequently we have $\lim_{\lambda\to+\infty}\psi(\lambda)<0$. Also, since we have $\lim_{\lambda\to 0}\psi(\lambda)\to+\infty$. It is also noted that $\psi(\lambda)$ is a decreasing function of $\lambda$. Therefore, the Equation $\psi(\lambda)=0$ has a unique positive solution, and thus the existence and uniqueness of the MLE of $\lambda$ is demonstrated.  
\section*{Appendix C. Elements of the FIM}\label{appendix-C}
The Fisher information matrix $\mathcal{V}(\theta)$ is given by 
\begin{equation*}
        \mathcal{V}(\theta) =
    \begin{pmatrix}
        \mathcal{V}_{11} & \mathcal{V}_{12}\\
        \mathcal{V}_{21} & \mathcal{V}_{22}
    \end{pmatrix}
\end{equation*}
and the expressions of the elements are as follows
\begin{align*}
    \mathcal{V}_{11} & = \frac{1}{\eta^2}\left[n\int_0^{X_0}f(x)dx + \sum_{i=1}^ri\binom{n}{i}\int_{F(X_0)}^1 x^{i-1}(1-x)^{n-i}du \right]\\
    \mathcal{V}_{12} & = \frac{1}{\eta}\left[n\int_0^{X_0}H(x)f(x)dx + \sum_{i=1}^ri\binom{n}{i}\int_{F(X_0)}^1 G(x)x^{i-1}(1-x)^{n-i}du \right]\\
    \mathcal{V}_{22} & = n\int_0^{X_0}(H(x))^2f(x)dx + \sum_{i=1}^ri\binom{n}{i}\int_{F(X_0)}^1 (G(x))^2x^{i-1}(1-x)^{n-i}du 
\end{align*}
where 
$$H(x)=\frac{1}{\lambda}+(1+x^\lambda)\ln{x},$$
and $$G(x)=\frac{1}{\lambda}\left[1 + \left\{1 + \ln{\left(1-\frac{\ln{(1-x)}}{\eta}\right)} \right\}\ln\left(\ln{\left(1-\frac{\ln{(1-x)}}{\eta}\right)}\right)  \right].$$

\end{document}